%% file: main.tex
\documentclass[11pt]{article}

\usepackage{palatino}
\usepackage{mathpazo}
\usepackage{braket}
\usepackage{amsfonts}
\usepackage{amssymb}
\usepackage{amsmath}
\usepackage{latexsym}
\usepackage{amsthm}
\usepackage[usenames]{color}
\usepackage{hyperref}
\usepackage{optidef}
\usepackage{thmtools}
\usepackage{cleveref}
\usepackage{relsize}
\usepackage{graphicx}
\usepackage{tikz}
\usepackage{authblk}
\usepackage{color,graphicx}
\usepackage[normalem]{ulem}
\usepackage{dsfont}
\usepackage{algorithm}
\usepackage{algpseudocode}
\usepackage{enumitem}
\usepackage{booktabs}

\theoremstyle{definition}

\hypersetup{pdfpagemode=UseNone}

\newtheorem{theorem}{Theorem}[section]
\newtheorem{lemma}[theorem]{Lemma}

\newtheorem{proposition}[theorem]{Proposition}
\newtheorem{cor}[theorem]{Corollary}
\newtheorem{corollary}[theorem]{Corollary}
\theoremstyle{definition}

\newtheorem{remark}[theorem]{Remark}

\newtheorem{claim}[theorem]{Claim}

\newtheorem*{theorem*}{Theorem}

\newcommand{\C}{\mathbb{C}}

\newcommand{\A}{\mathcal{A}}

\newcommand{\Herm}{\textup{Herm}} 
\newcommand{\D}{\textup{D}} 

\newcommand{\tr}{\operatorname{tr}}
\newcommand{\Tr}{\operatorname{tr}}

\newcommand{\kb}[1]{\ket{#1}\bra{#1}}

\newcommand{\norm}[1]{\left\lVert#1\right\rVert}

\newcommand{\Dens}{\textup{D}}

\newcommand{\ip}[2]{\left\langle #1,#2\right\rangle}
\newcommand{\abs}[1]{\left\lvert #1\right\rvert}

\begin{document}

\title{\bf Online learning of quantum states under structure}

\author[1]{Akshay Bansal\footnote{\href{mailto:akshay.bansal@tuwien.ac.at}{\texttt{akshay.bansal@tuwien.ac.at}}}}
\author[2]{Jiahui Liu\footnote{\href{mailto:jliu@fujitsu.com}{\texttt{jliu@fujitsu.com}}}}

\affil[1]{Technische Universit\"at Wien, Austria}
\affil[2]{Fujitsu Research of America, USA}

\date{\today}

\maketitle

\begin{abstract}
Quantum state tomography is fundamental to quantum information processing but becomes infeasible at scale due to the exponential growth of the state space. Shadow tomography alleviates this challenge by focusing on predicting measurement outcomes rather than reconstructing the full state. Its online variant models adaptive and potentially adversarial measurement scenarios, where a learner sequentially predicts outcomes while competing with the best fixed quantum state in hindsight.

We show that exploiting additional structure in the measurements leads to significantly stronger regret guarantees. In particular, under the assumption that the adversarial measurements have bounded Frobenius norm, we analyze online mirror descent and derive optimal regret bounds that depend on intrinsic structural properties, such as rank or sparsity in the standard basis, rather than the dimension of the measurement operators.
As a complementary result, we also show that, even in the setting where adversarial measurements are known to be sparse in the Pauli basis commonly used in variational quantum eigensolvers and near-term quantum error mitigation, the underlying regret bound for learning quantum states is the same as that obtained in the generic setting, where the adversarial measurements are not known to possess any particular structure.
\end{abstract}


\input{introduction}

\input{technical_results}

\input{prelim}

\section{Tight minimax regret for bounded effects}\label{sec:minmax-bounded-effects}
\input{new_upper_bound}
\input{lower_bound}

\input{acks}
\input{ai_disclosure}
\bibliographystyle{alpha}
\bibliography{references}

\appendix

\input{ogd_regret}
\input{sparse_lower}

\end{document}

%% file: introduction.tex
\section{Introduction}\label{sec:introduction}

State tomography is a fundamental task in quantum computation and quantum information, essential for both theoretical research and practical applications in a number of disciplines. 
It aims at finding a classical description of an unknown quantum state (referred to as lab state), given a finite number of copies of the state. 
State tomography plays a crucial role in verifying the outcomes of quantum algorithms, evaluating quantum error correction codes, and assessing quantum states generated in cryptographic protocols, including but not limited to experiments in quantum error correction, quantum key distribution as well as validating correctness and performance in quantum chemistry algorithms.

However, general state tomography is highly inefficient, not only computationally but also in terms of number of samples one needs, 
due to the intrinsic challenge of learning a complete description of an arbitrary, unknown quantum state -- classically described by a matrix that scales exponentially in the number of qubits. The prior work by~\cite{haah2016}  established an exponential (in the number of qubits) lower bound on the number of copies required for this process when the lab state has no known structure. The same work also provided a measurement scheme that matches this lower bound and thus making it sample optimal.

To handle the challenges in general tomography in real-world experiments, \cite{aaronson2017shadow} introduced the concept of shadow tomography. Instead of reconstructing the full density matrix of a quantum state close in trace distance, shadow tomography focuses on extracting specific properties of the target state. Specifically, given multiple copies of an unknown quantum state and a set of 2-outcome measurements (known apriori), the goal is to produce a representative quantum state such that the measurement outcome distributions for the representative state closely approximate those of the unknown quantum state.

Shadow tomography significantly reduces the sample complexity required for learning quantum states. However, it comes with certain practical limitations. Firstly, it does not account for adversarial environments or those that evolve over time as all the measurements are considered to be drawn independently from a fixed distribution.
Furthermore, not all quantum measurements are accessible or feasible in a given experimental setup, as their feasibility is dictated by physical constraints. As experimental control advances, more sophisticated measurement techniques become available. To address these considerations, \cite{aaronsononline} formalized a notion that incorporates real-time measurement constraints into tomography, framing the problem as an instance of \emph{online learning} of quantum states where they analyzed several state-of-the-art online algorithms such as Regularized Follow-the-Leader (RFTL) and the Matrix Multiplicative Weights (MMW) method.

\paragraph{Applications of online learning in quantum.}
In quantum computing, there are a wide range of applications that can be modeled within the online framework including but not limited to simulation of time-varying quantum channels and adaptive control of qubits, especially in noisy intermediate-scale quantum (NISQ) devices~\cite{chittoor23}.
Methods such as reinforcement learning and neural networks can optimize gate operations, learn feedback control policies, dynamically updating weights for quantum neural nets, and reconstruct quantum states or processes from streaming measurement data~\cite{fosel2021reinforcement,miao2024neural, gaikwad2023neural}. In quantum chemistry, these techniques are used to enhance variational algorithms for estimating molecular ground-state energies, allowing for more efficient simulations of chemical systems~\cite{ghosh2023deep}. In materials science, online learning can help design control protocols for synthesizing quantum materials or navigating complex phase spaces in many-body systems~\cite{rajak2020deep}. Overall, these methods can potentially replace static, model-based procedures with adaptive, data-driven approaches that can continuously learn and improve from ongoing experiments or simulations.

While the seminal work~\cite{aaronsononline} considered a much general setting of online learning of quantum states, many practical applications admit additional structural properties that can be exploited to obtain improved regret bounds in terms of the number of qubits and time horizon.

Given that the setting of online learning of quantum states under certain adversarial structures (with practical implications in quantum computing we will discuss in ~\ref{sec:directions_contributions}) remains largely unexplored in the current literature, we investigate whether additional assumptions on the measurement operators can lead to stronger regret guarantees. Informally, we address the following question:
\begin{center}
    \emph{Can tight regret bounds be obtained when the adversarial measurement operators possess specific structural properties, such as sparsity or low rank?}
\end{center}

\subsection{Our directions and contributions}
\label{sec:directions_contributions}

We next discuss the main contributions of our work with the formal statements presented in~\Cref{sec:technical-results}.

\emph{As a first result, we show that when the adversarial effects are known to be low-rank or sparse in the standard basis (or, more generally, bounded in Frobenius norm), the regret scales almost logarithmically in their rank or sparsity, with this scaling being optimal.}

This result not only recovers the regret bounds established in~\cite{aaronsononline} but also provides a strict improvement whenever the effects are, for instance, low-rank or sparse. In fact, a large number of practical applications exhibit such structure, which can be leveraged to achieve better performance and learnability.

For example, it is now well established that \emph{optimal} full-state tomography can be achieved using symmetric informationally complete (SIC) POVMs~\cite{ohno2015examples} (rank-one operators), and this has also been demonstrated experimentally on an ion-trap quantum processor~\cite{stricker2022experimental}. 
As another example, quantum error-correction experiments repeatedly measure stabilizer syndromes to determine whether errors have occurred. A joint syndrome outcome corresponds to a projector onto a syndrome subspace; for an \(n\)-qubit stabilizer code encoding \(k\) logical qubits, this projector has rank \(2^k\), which can be still be far smaller when many stabilizers are measured~\cite{le2023high, mcintyre2024}. Repeated syndrome measurements therefore give a very natural online stream of structured effects, and recent online-learning work explicitly identifies parity/syndrome-type measurements in quantum-error-correction settings as an important source of low-rank structure.

Sparsity in the standard basis arises when a search~\cite{figgatt2017complete} or verification test~\cite{zhang2020experimental} accepts only a small set of computational-basis outcomes. Pauli sparsity also arises in quantum error correction, where repeatedly measuring a small number of stabilizer checks produces a stream of effects involving only a few Pauli terms~\cite{bausch2024learning,hall2024artificial}.

\emph{Our results show that sparsity in the standard basis can improve regret, whereas Pauli sparsity alone does not guarantee an improvement.}

Thus, the measurement operators (or effects) chosen by the adversary need not be arbitrary full-dimensional operators; instead, they may be low-rank, sparse in a preferred basis, or supported on a physically meaningful truncated subspace.

%% file: technical_results.tex
\section{Technical results}\label{sec:technical-results}

\subsection{Background: Online learning of quantum states}\label{sec:background}
Before we go into the technical results, we formally give the background and definition of quantum state online learning.

Consider an online interaction between a learner and an adversary that spans over $T$ number of discrete time steps. At each time step $t \in [T]$, learner reveals a quantum state $\rho_t$ to the adversary by providing its full classical description. The adversary responds by selecting a measurement operator $E_t$ and a loss function $\ell_t$, and subsequently revealing the full description of $E_t$ and $\ell_t$ to the learner. The learner experiences the loss $\ell_t(\Tr(E_t\rho_t))$ and updates its guess to $\rho_{t+1}$ for the next time step.  Here, $E_t$ is positive semidefinite satisfying $E_t \preccurlyeq \mathds{1}$ (with $\mathds{1}$ being the identity operator), $\Tr$ is used to denote the operator trace. As usual in the online setting, the learner wishes to devise a sequence of guess estimates $\{\rho_t\}_{t \in [T]}$ that minimizes its overall regret given by:
\begin{equation}\label{eq:regretStateLearning}
    R_{T} \coloneqq \max_{\ell_1\ldots, \ell_T} \Big( \sum_{t \in [T]} \ell_t\big(\Tr(E_t\rho_t)\big) - \inf_{\rho \in \D(\mathbb{C}^{n})}\sum_{t \in [T]} \ell_t\big(\Tr(E_t\rho)\big)  \Big) \enspace,
\end{equation}
where $\D(\mathbb{C}^{2^n})$ is the set of all $n$-qubit quantum states of dimension $2^{n} \times 2^{n}$ and $\rho_t \in \D(\mathbb{C}^{2^n})$ for all $t \in [T]$.

\input{related_works_applications}

\subsection{Our results}

We strengthen the general regret guarantees established in~\cite{aaronsononline} by incorporating additional structural assumptions on the measurement operators. 

Our first result concerns the setting in which the adversarial measurement operators have bounded Frobenius norm and is presented next.

\begin{theorem}[Tight minimax regret for bounded effects]\label{thm:regret-bounded-norm}
Let $E_1, E_2, \ldots, E_T$ be a sequence of two-outcome measurements on an $n$-qubit state such that their Frobenius norms satisfy $\lVert E_t \rVert_F \leq B$ for all $t \in [T]$. Let $\ell_1, \ell_2, \ldots, \ell_T$ be a sequence of loss functions that are convex and $L$-Lipschitz, so that, for every $t \in [T]$ and every $x \in \mathbb{R}$, there exists a sub-derivative $\ell'_t(x)$ satisfying $\lvert \ell'_t(x) \rvert \leq L$. Then, under the online learning setting, where the measurement $E_t$ and loss function $\ell_t$ are chosen adversarially, there exists an optimal learning strategy that guarantees 
\[ 
R_T = \Theta\left(L\sqrt{T \ln(1 + B^2)}\right).
\]
\end{theorem}
A proof of the previous theorem is detailed in~\Cref{sec:minmax-bounded-effects}. 

\begin{remark}
    Note that the previous result already recovers the regret bound given in~\cite[Theorem 2]{aaronsononline} as $0 \preccurlyeq E_t \preccurlyeq \mathbb{1}$ trivially implies $B \leq 2^n$. Moreover, \Cref{thm:regret-bounded-norm} also establishes the optimality of the stated regret whenever the loss functions are known to be convex and Lipshitz. This also settles the open question posed in~\cite[Section 6]{aaronsononline} on the regret optimality for $L_2$-loss functions since $L_2$ losses are just $1$-Lipschitz.
\end{remark}

\begin{remark}
    Previously, \cite{chen2006more} considered a similar online learning setting in which the Frobenius norm of the effects was known to be bounded by some constant $B$. However, their regret analysis, based on the Regularized Follow-the-Leader approach with \emph{Tsallis-2} entropy as the regularizer, resulted in a regret bound of $\mathcal{O}(LB\sqrt{T})$. In this work, we significantly improve this bound in the form of~\Cref{thm:regret-bounded-norm} via a different online approach.
\end{remark}

Continuing from the last remark, one can recover the regret bounds given in~\cite{chen2006more} using only the online gradient descent (OGD) approach, which is comparatively easy to analyze in this problem setting. Moreover, the successive updates due to OGD can be evaluated much more simply, without the need to solve any semidefinite program, as is required by the Regularized Follow-the-Leader approach with \emph{Tsallis-2} entropy used in~\cite{chen2006more}. We present this finding as a proposition below, since the regret bound in~\Cref{thm:regret-bounded-norm} is still substantially stronger.

\begin{proposition}
Let $T$ denote the time horizon and $n$ the number of qubits. Suppose that, at each round $t \in [T]$, the loss function is convex and $L$-Lipschitz. Further, assume that the adversarial measurement operators satisfy $\norm{E_t}_{F} \leq B$ for all $t \in [T]$. Then, the regret of Projected Online Gradient Descent (OGD) satisfies 
\[
R_T \leq \mathcal{O}\big(LB\sqrt{T}\big).
\]

Moreover, each successive guesstimate update can be computed with classical computational complexity $\mathcal{O}(2^{\omega n})$, where $2 \leq \omega \leq 2.37$ is the matrix multiplication exponent.
\end{proposition}
We outline a proof of this proposition in~\Cref{sec:regret-ogd}. 

\begin{remark}
    When \(0<B\ll1\), \(\ln(1+B^2)=\Theta(B^2)\), so~\Cref{thm:regret-bounded-norm} reduces to \(\Theta(LB\sqrt T)\).  Thus the Projected OGD result is tight in the small-budget regime.  When \(B\gg1\), however, the correct dependence is only \(\sqrt{\ln(1+B^2)}\), not \(B\).
\end{remark}

The following two results are direct corollaries due to~\Cref{thm:regret-bounded-norm} and might be of independent interest in quantum information and learning theory.

\begin{corollary}[Tight minimax regret for low-rank effects]\label{corr:regret-low-rank}
If the adversarial effects satisfy $\operatorname{rank}(E_t) \leq r$ for all $t \in [T]$, then
\[ 
R_T = \Theta\left(L\sqrt{T \ln(1 + r)}\right).
\]
\end{corollary}

The above holds since $0 \preccurlyeq E_t \preccurlyeq \mathbb{1}$ and, thus, $\lVert E_t \rVert_F^2 = \operatorname{rank}(E_t) \leq r$. Similarly, if the effects $E_t$ are known to be $\kappa$-sparse in the standard basis, then $\operatorname{rank}(E_t) \leq \kappa$, which implies the following corollary.
\begin{corollary}[Tight minimax regret for sparse effects]\label{corr:regret-sparse-comp}
If the adversarial effects are $\kappa$-sparse in the standard basis for all $t \in [T]$, then
\[
R_T = \Theta\left(L\sqrt{T \ln(1 + \kappa)}\right).
\]
\end{corollary}

\begin{table}[t]
\centering
\begin{tabular}{@{}lll@{}}
\toprule
Admissible effects & Previous upper bound & Tight minimax regret \\
\midrule
Arbitrary, dimension \(d\) & \(O(L\sqrt{nT})\)~\cite{aaronsononline} & \(\Theta(L\sqrt{nT})\)\textbf{(this work)} \\
Rank at most \(r\) & \(O(L\sqrt{rT})\)~\cite{chen2006more} & \(\Theta(L\sqrt{T\ln(1+r)})\)\textbf{(this work)} \\
Frobenius norm at most \(B\) & \(O(LB\sqrt T)\)~\cite{chen2006more} & \(\Theta(L\sqrt{T\ln(1+B^2)})\)\textbf{(this work)} \\
\bottomrule
\end{tabular}
\caption{Comparison of minimax regret bounds.}
\label{tab:regret-comparison}
\end{table}

As a complementary result, we also consider the online setting for learning quantum states when the adversarial effects are instead known to be sparse in the Pauli basis. Contrary to~\Cref{corr:regret-sparse-comp}, which demonstrates improved regret bounds for effects that are sparse in the standard basis, a straightforward lower-bound analysis establishes the impossibility of any such improvement when the effects are only known to be sparse in the Pauli basis, as shown in the following theorem.

\begin{proposition}[Tight minimax regret for Pauli sparse effects]
    Let $E_1, E_2, \ldots, E_T$ be a sequence of two-outcome measurements on an $n$-qubit state such that the effects $E_t$ are $\kappa$-sparse in Pauli basis for all $t \in [T]$. Let $\ell_1, \ell_2, \ldots, \ell_T$ be a sequence of loss functions that are convex and $L$-Lipschitz, so that, for every $t \in [T]$ and every $x \in \mathbb{R}$, there exists a sub-derivative $\ell'_t(x)$ satisfying $\lvert \ell'_t(x) \rvert \leq L$. Then, under the online learning setting, where the measurement $E_t$ and loss function $\ell_t$ are chosen adversarially, there exists an optimal learning strategy that guarantees 
\[ 
R_T = \Theta\left(L\sqrt{nT}\right).
\]
\end{proposition}
A detailed proof of this proposition is given in~\Cref{sec:regret-sparse-effects}.

%% file: related_works_applications.tex
\paragraph{Related works.}
Since the seminal work by~\cite{aaronsononline}, there have been several other works on the related topic. In~\cite{yang2020revisiting}, the authors analyzed a simple variant of Regularized Follow-the-Leader (RFTL) in the form of online gradient descent (OGD) to recover regret bounds that still scale sublinearly in the time $T$. Their regret bounds, however, contained factors that scale exponentially in the number of qubits ($n$). A related work~\cite{tseng2025} considered logarithmic loss functions, which violate the Lipschitz-continuity assumption used in~\cite{aaronsononline}, and demonstrated regret bounds that scale in $\log(T)$, a substantial improvement over the $\sqrt{T}$ bounds known previously. Notably, however, these bounds were achieved under a logarithmic loss function that is known to the learner a priori.

More recently, \cite{bansal2025online} extended the online framework to learn arbitrary quantum objects that are known to be positive semidefinite with bounded trace, establishing regret bounds that scale sublinearly in time $T$. Another work~\cite{raza2024online} considered a related problem of learning quantum processes and achieved significantly improved regret bounds in terms of the number of qubits ($n$), while considering quantum processes with bounded gate complexity as well as Pauli channels.

In the problem of offline shadow tomography where the task is to approximate $\tr(E \rho)$ (here $E$ is a given measurement operator and $\rho$ is an unknown quantum state), one would naturally expect to accomplish this task in fewer samples if some additional information about $\rho$ (for example, purity) is known. In fact, the problem has since been studied in~\cite{grier2024sample} where the authors demonstrated a sample-optimal complexity strictly better than the one known for general quantum states~\cite{aaronson2017shadow} (also see~\cite{sinha2025dimension,jeronimo2026dimension}).
However, in the online setting, it was noted in~\cite{meyer2026online} that restricting the learning task to the set of pure states does not automatically yield improved regret guarantees when compared with online learning over general quantum states. 

In another work, \cite{chen2006more} analyzed the regret behavior for online learning of quantum states where the adversarial effects are known to have bounded Frobenius norm, yielding regret guarantees that are strictly better than those discussed in~\cite{aaronsononline} under some conditions. 


While much of the literature in online learning of quantum states discussed only static regret (fixed best comparator), \cite{Chen2024adaptiveonline} extended the analysis for adaptive and dynamic regret while establishing their scaling in just polynomial number of qubits. 

%% file: prelim.tex
\section{Preliminaries and notation}

We introduce the notation and basic definitions used throughout the paper. The complex Euclidean space of dimension $d$ is denoted by $\mathbb{C}^d$. For an $n$-qubit state, the dimension of the state space is $d=2^n$, and its standard (or computational) basis, in Dirac notation, is
\[
\{\ket{x}\}_{x\in\{0,1\}^n}.
\]
For example, the standard basis of a two-qubit state space is $\{\ket{00},\ket{01},\ket{10},\ket{11}\}$.

For a complex Euclidean space $\mathcal{X}$, the identity operator acting on $\mathcal{X}$ is denoted by $\mathbb{1}_{\mathcal{X}}$, or simply by $\mathbb{1}$ when the underlying space is clear from context. The set of Hermitian operators acting on $\mathcal{X}$ is denoted by $\textup{Herm}(\mathcal{X})$; that is, $X\in\textup{Herm}(\mathcal{X})$ if $X^*=X$. The set of positive semidefinite operators with unit trace is denoted by $\textup{D}(\mathcal{X})$ (or $\textup{D}_{\dim(\mathcal{X})}$) and corresponds to the set of quantum states (or density operators) on $\mathcal{X}$. For $P,Q\in\textup{Herm}(\mathcal{X})$, the Hilbert--Schmidt inner product is defined as
\begin{equation*}
    \langle P,Q\rangle=\Tr(PQ).
\end{equation*}

\subsection{States, effects, and norms}

Fix an integer $d\geq 2$. The learner's decision set is the set of density operators on $\C^d$,
\[
  \Dens_d
  :=
  \left\{\rho\in\Herm(\C^d):\rho\succeq 0,\ \Tr(\rho)=1\right\}.
\]
An \emph{effect} or \emph{measurement operator} is a Hermitian operator $E$ satisfying
\[
0\preceq E\preceq\mathbb{1}.
\]
Consequently, every eigenvalue of $E$ lies in $[0,1]$, and $\Tr(E\rho)\in[0,1]$ for every $\rho\in\Dens_d$.

For $1\leq p<\infty$, the Schatten-$p$ norm of a matrix $X$ is defined by
\[
  \norm{X}_{S_p}
  :=
  \bigl(\Tr\abs{X}^{p}\bigr)^{1/p},
\]
where $\abs{X}=(X^*X)^{1/2}$. For $p=2$, the Schatten-$2$ norm coincides with the Frobenius (or Hilbert--Schmidt) norm:
\[
\norm{X}_F=\norm{X}_{S_2}.
\]
For conjugate exponents $p,q$ satisfying $1/p+1/q=1$, the Schatten-$p$ and Schatten-$q$ norms are dual to each other, and the Schatten--H\"older inequality gives
\begin{equation}\label{eq:schatten-holder}
    \abs{\Tr(XY)}
  \leq
  \norm{X}_{S_p}\norm{Y}_{S_q}.
\end{equation}

\subsection{Divergence inequality}

Let $\mathcal{K}$ be a convex subset of a real inner-product space $\mathcal{X}$, and let $F:\mathcal{K}\to\mathbb{R}$ be a differentiable convex function on interior of $\mathcal{K}$ ($\operatorname{int}(\mathcal{K})$). The Bregman divergence induced by $F$ is defined, for $P\in\mathcal{K}$ and $Q\in\operatorname{int}(\mathcal{K})$, as
\begin{equation}\label{eq:def-bregman}
    D_F(P, Q)
    :=
    F(P)-F(Q)-\langle \nabla F(Q),P-Q\rangle.
\end{equation}
In particular, the Bregman divergence can be viewed as the error in the first-order Taylor approximation of $F$ about $Q$ when evaluated at $P$.

Fix $p\geq 2$ and let
\begin{equation}
    q:=\frac{p}{p-1}\in(1,2].
\end{equation}
For Hermitian operators, we define the Schatten regularizer
\begin{equation}
    \Psi_q(X):=\frac{1}{2}\|X\|_{S_q}^2.
\end{equation}

The following sharp uniform-convexity fact is due to Ball, Carlen, and Lieb~\cite{ballcarlenlieb1994}.
\begin{lemma}[Strong convexity of squared Schatten norms]\label{lemma:ball-carlen-lieb}
For $1 < q \leq 2$, the function
\[
    \Psi_q(X) = \frac{1}{2}\lVert X \rVert_{S_q}^2
\]
is $(q-1)$-strongly convex with respect to the $\lVert \cdot \rVert_{S_q}$ norm. Equivalently, its Bregman divergence
\[
    D_{\Psi_q}(X,Y)
    := \Psi_q(X)-\Psi_q(Y)
    - \langle \nabla \Psi_q(Y), X-Y\rangle
\]
satisfies
\[
    D_{\Psi_q}(X,Y)
    \geq
    \frac{q-1}{2}\lVert X-Y\rVert_{S_q}^2.
\]
\end{lemma}

Recall that the inner product is the trace inner product
\[
    \langle X,Y\rangle = \Tr(XY)
\]
on Hermitian matrices. The Ball--Carlen--Lieb theorem is stated for the full matrix space; restricting it to the real vector space of complex Hermitian matrices preserves the same strong-convexity constant.

%% file: new_upper_bound.tex
\subsection{Upper bound by Schatten online mirror descent}

To evaluate an upper bound on the regret for bounded effects, we analyze~\Cref{alg:schatten-md} based on an online mirror descent approach. For a more pedagogical discussion on various online learning algorithms, please refer to~\cite{hazan2016, orabona2023modernintroductiononlinelearning}. 

\begin{algorithm}[H]
\caption{Schatten mirror descent approach for learning quantum states}
\label{alg:schatten-md}
\begin{algorithmic}[1]
\Require No. of qubits $n$, horizon $T$, Schatten parameter $1 \leq p \leq \infty$, and Lipschitz constant $L > 0$
\State Set
\[
    d 
    := 
    2^n,
    \qquad
    q
    :=
    \frac{p}{p-1}
\]

\State Define the Schatten-$q$ regularizer
\[
    \Psi_q(\rho)
    :=
    \frac12\|\rho\|_{S_q}^{\,2},
    \qquad
    \|\rho\|_{S_q}
    :=
    \bigl(\Tr|\rho|^q\bigr)^{1/q},
\]
and its Bregman divergence
\[
    D_{\Psi_q}(\rho,\sigma)
    :=
    \Psi_q(\rho)-\Psi_q(\sigma)
    -
    \left\langle
        \nabla\Psi_q(\sigma),\rho-\sigma
    \right\rangle.
\]
\State Initialize
\[
    \rho_1:=\frac{\mathbb{1}}{d}.
\]
\For{$t=1,\ldots,T$}
    \State Output the density matrix $\rho_t$.
    \State Observe the effect $0\preceq E_t\preceq \mathbb{1}$ and the convex
    $L$-Lipschitz loss function $\ell_t:[0,1]\to\mathbb R$.
    \State Compute the scalar prediction
    \[
        z_t:=\tr(E_t\rho_t).
    \]
    \State Choose a scalar subgradient $\gamma_t\in[-L,L]$ satisfying
    \[
        \ell_t(z)
        \ge
        \ell_t(z_t)+\gamma_t(z-z_t)
        \qquad
        \text{for every }z\in[0,1].
    \]
    \State Define the Hermitian loss subgradient
    \[
        G_t:=\gamma_t E_t.
    \]
    \State Update
    \begin{equation}\label{eq:omd-update}
        \rho_{t+1}
        :=
        \underset{\rho\in \textup{D}_d}{\operatorname{argmin}}
        \left\{
            \eta\langle G_t,\rho\rangle
            +
            D_{\Psi_q}(\rho,\rho_t)
        \right\}.
    \end{equation}
\EndFor
\end{algorithmic}
\end{algorithm}

For every comparator \(\rho\in\Dens_d\), scalar convexity gives
\begin{align}
  \ell_t\bigl(\Tr(E_t\rho_t)\bigr)
  -\ell_t\bigl(\Tr(E_t\rho)\bigr)
  &\leq
  a_t\Tr\bigl(E_t(\rho_t-\rho)\bigr)\notag\\
  &=\Tr\bigl(G_t(\rho_t-\rho)\bigr).
  \label{eq:linearization}
\end{align}
It therefore suffices to control online linear regret for the gradients \(G_t\).

We next show the following result that will be used in the final regret analysis.
\begin{lemma}\label{lem:initial-divergence}
For every \(\rho\in\Dens_d\),
\begin{equation}\label{eq:initial-divergence}
  D_{\Psi_q}\!\left(\rho,\frac{\mathbb{1}}{d}\right)\leq\frac12.
\end{equation}
\end{lemma}

\begin{proof}
For positive semidefinite \(X\neq0\), differentiation under the trace gives
\[
  \nabla\Psi_q(X)=\norm X_{S_q}^{2-q}X^{q-1}.
\]
At \(X=\mathbb{1}/d\), this gradient is a scalar multiple of the identity.  Hence
\[
  \ip{\nabla\Psi_q(\mathbb{1}/d)}{\rho-\mathbb{1}/d}=0
\]
because both \(\rho\) and \(\mathbb{1}/d\) have trace one.  It follows that
\[
  D_{\Psi_q}(\rho,\mathbb{1}/d)
  =\Psi_q(\rho)-\Psi_q(\mathbb{1}/d)
  \leq\Psi_q(\rho).
\]
The eigenvalues of \(\rho\) are nonnegative and sum to one, so \(\norm\rho_{S_q}\leq\norm\rho_{S_1}=1\).  Therefore \(\Psi_q(\rho)\leq1/2\).
\end{proof}

\begin{lemma}[OMD inequality]\label{lem:omd-generic}
For the update in~\Cref{eq:omd-update}, every comparator \(\rho\in\Dens_d\) satisfies
\begin{equation}\label{eq:omd-generic}
  \sum_{t=1}^T\Tr\bigl(G_t(\rho_t-\rho)\bigr)
  \leq
  \frac{D_{\Psi_q}(\rho,\rho_1)}{\eta}
  +
  \frac{\eta}{2(q-1)}
  \sum_{t=1}^T\norm{G_t}_{S_p}^2.
\end{equation}
\end{lemma}

\begin{proof}
The first-order optimality condition for~\Cref{eq:omd-update} states that, for every \(\rho\in\Dens_d\),
\[
  \ip{\eta G_t+\nabla\Psi_q(\rho_{t+1})-\nabla\Psi_q(\rho_t)}
     {\rho-\rho_{t+1}}
  \geq0.
\]
Using the definition of Bregman divergences (\Cref{eq:def-bregman}), this equivalently implies
\begin{equation}\label{eq:three-point-step}
  \eta\Tr\bigl(G_t(\rho_{t+1}-\rho)\bigr)
  \leq
  D_{\Psi_q}(\rho,\rho_t)
  -D_{\Psi_q}(\rho,\rho_{t+1})
  -D_{\Psi_q}(\rho_{t+1},\rho_t).
\end{equation}
Add \(\eta\Tr(G_t(\rho_t-\rho_{t+1}))\) to both sides yields,
\begin{equation}\label{eq:master-omd}
  \eta\Tr\bigl(G_t(\rho_{t}-\rho)\bigr)
  \leq
  D_{\Psi_q}(\rho,\rho_t)
  -D_{\Psi_q}(\rho,\rho_{t+1})
   + \eta\Tr(G_t(\rho_t-\rho_{t+1})) -D_{\Psi_q}(\rho_{t+1},\rho_t).
\end{equation}

Now,
\begin{align*}
  &\eta\Tr\bigl(G_t(\rho_t-\rho_{t+1})\bigr)
  -D_{\Psi_q}(\rho_{t+1},\rho_t)\\
  &\qquad\leq
  \eta\norm{G_t}_{S_p}\norm{\rho_t-\rho_{t+1}}_{S_q}
  -\frac{q-1}{2}\norm{\rho_t-\rho_{t+1}}_{S_q}^2\\
  &\qquad\leq
  \frac{\eta^2}{2(q-1)}\norm{G_t}_{S_p}^2,
\end{align*}
where the first inequality is due to~\Cref{eq:schatten-holder} and~\Cref{lemma:ball-carlen-lieb}, while the second equality follows from the fact that \(uv-\frac{q-1}{2}v^2\leq\frac{u^2}{2(q-1)}\) with $u := \eta \norm{G_t}_{S_p}$ and $v := \norm{\rho_t-\rho_{t+1}}_{S_q}$.

Substituting into~\Cref{eq:master-omd}, followed by summing over \(t\), telescopes the two comparator-divergence terms.  Dropping the final nonnegative divergence and dividing by \(\eta\) proves~\Cref{eq:omd-generic}.
\end{proof}

Finally, combining~\Cref{eq:linearization,eq:omd-generic} and~\Cref{lem:initial-divergence} yields
\begin{equation}\label{eq:pre-optimize}
  R_T
  \leq
  \frac{1}{2\eta}
  +\frac{\eta}{2(q-1)}L^2
  \sum_{t=1}^T\norm{E_t}_{S_p}^2.
\end{equation}
Because \(q-1=1/(p-1)\), optimizing \(\eta\) gives the following master bound.

\begin{proposition}[Schatten master bound]\label{prop:schatten-master}
Fix \(p\geq2\), and suppose a bound \(M_p\geq0\) satisfying
\(\norm{E_t}_{S_p}\leq M_p\) on every round is known to the learner.  Then online mirror descent with
\begin{equation}\label{eq:master-eta}
  \eta:=\frac{1}{LM_p\sqrt{(p-1)T}}
\end{equation}
when \(LM_p>0\) satisfies
\begin{equation}\label{eq:schatten-master-uniform}
  R_T\leq LM_p\sqrt{(p-1)T}.
\end{equation}
If \(LM_p=0\), the regret is zero.
\end{proposition}

\begin{proof}
Because \(q-1=1/(p-1)\), the uniform bound on the effects turns~\Cref{eq:pre-optimize} into
\[
  R_T
  \leq
  \frac{1}{2\eta}
  +\frac{\eta L^2M_p^2(p-1)T}{2}.
\]
For \(LM_p>0\), substituting~\Cref{eq:master-eta} makes the two terms equal and gives~\Cref{eq:schatten-master-uniform}.  If \(L=0\), every loss is constant.  If \(M_p=0\), every admissible effect is zero.  In either case the regret is zero.  Thus the learning rate depends only on the known horizon and class parameters, not on future effects.
\end{proof}

Suppose \(0\preceq E\preceq\mathbb{1}\) and \(\norm E_F\leq B\).  Since every eigenvalue lies in \([0,1]\), for \(p\geq2\),
\begin{equation}\label{eq:frob-to-sp}
  \norm E_{S_p}^p
  =\sum_i\lambda_i(E)^p
  \leq\sum_i\lambda_i(E)^2
  =\norm E_F^2
  \leq B^2.
\end{equation}
Thus \(\norm E_{S_p}\leq B^{2/p}\).

If \(0<B\leq1\), use \(p=2\) and \(M_2=B\) in~\Cref{prop:schatten-master} to obtain
\begin{equation}\label{eq:frob-upper-small}
  R_T\leq LB\sqrt T.
\end{equation}
Since \(\ln(1+B^2)\geq B^2/2\) in this range,
\[
  LB\sqrt T\leq\sqrt2\,L\sqrt{T\ln(1+B^2)}.
\]

If \(B>1\), choose
\begin{equation}\label{eq:frob-p-choice}
  p:=2(1+\ln B^2).
\end{equation}
Then
\[
  B^{2/p}
  =\exp\!\left(\frac{\ln B^2}{2(1+\ln B^2)}\right)
  \leq\sqrt e,
  \qquad
  p-1\leq4\ln(1+B^2).
\]
Consequently, applying~\Cref{prop:schatten-master} with \(M_p=B^{2/p}\) gives
\begin{equation}\label{eq:frob-upper-large}
  R_T\leq2\sqrt e\,L\sqrt{T\ln(1+B^2)}.
\end{equation}
Combining~\Cref{eq:frob-upper-small,eq:frob-upper-large} proves the upper bound in~\Cref{thm:regret-bounded-norm}.

%% file: lower_bound.tex
\subsection{Lower bounds}
\label{sec:lower-bounds}

Aaronson et al.~\cite[
    Section~1, discussion following Theorem~2
]{aaronsononline}
recall the $\Omega(\sqrt{nT})$ absolute-loss lower bound in the
non-realizable setting, citing the classical experts lower bound
of Arora, Hazan, and Kale~\cite[
    Theorem~4.1
]{arora2012multiplicative}.
The latter uses a probabilistic argument in which the learner
incurs the mean cost, while the best fixed expert benefits from
random fluctuations in hindsight. We adapt this idea using computational-basis states labeled by binary strings. In each block of rounds, a diagonal measurement queries one bit of the string, and the loss compares the prediction with an independent, uniformly random binary label. The learner cannot anticipate these labels, whereas a single basis state chosen in hindsight can encode the majority label in every block. We construct the measurement effects on the subspace spanned by these basis states and set them to zero on all remaining basis states. Embedding the construction in the full quantum state space therefore does not increase the rank, Frobenius norm, or number of nonzero entries of any effect.

We detail this
adaptation below to make the structural restrictions explicit.

Recall that we use the unnormalized Frobenius norm
$\norm{M}_F^2=\Tr(M^\dagger M)$ and write
$\operatorname{nnz}(M)$ for the total number of nonzero matrix
entries in the computational basis.

\begin{cor}[Structural lower bounds]
\label[corollary]{cor:aaronson-structured-lower-bounds}
There is a universal constant $c>0$ such that, for every possibly
randomized learner $\A$, each of the following bounds holds on
some deterministic sequence of diagonal effects and convex
$L$-Lipschitz absolute losses.
The expectation is over the learner's randomness.

\smallskip
\noindent\textbf{Bounded rank.}
For every integer $1\leq r\leq d$, there is such a sequence with
$\operatorname{rank}(E_t)\leq r$ for all $t$ and
\begin{equation}
\label{eq:aaronson-rank-lower-bound}
    \mathbb{E}_{\A}[R_T]
    \geq
    cL\min\left\{
        T,\sqrt{T\ln(1+r)}
    \right\}.
\end{equation}

\smallskip
\noindent\textbf{Bounded Frobenius norm.}
For every $B>0$, there is such a sequence with
$\norm{E_t}_F\leq B$ for all $t$ and
\begin{equation}
\label{eq:aaronson-frobenius-lower-bound}
    \mathbb{E}_{\A}[R_T]
    \geq
    cL\min\left\{
        T,
        \sqrt{T\ln\bigl(1+\min\{B^2,d\}\bigr)}
    \right\}.
\end{equation}

\smallskip
\noindent\textbf{Total matrix-entry sparsity.}
For every integer $\kappa\geq1$, there is such a sequence with
$\operatorname{nnz}(E_t)\leq\kappa$ for all $t$ and
\begin{equation}
\label{eq:aaronson-sparsity-lower-bound}
    \mathbb{E}_{\A}[R_T]
    \geq
    cL\min\left\{
        T,
        \sqrt{T\ln\bigl(1+\min\{\kappa,d\}\bigr)}
    \right\}.
\end{equation}

The bounds hold even if the entire effect sequence is announced
before the interaction; the current loss is revealed only after
the learner's prediction.
\end{cor}

\begin{proof}
Fix a real parameter $s\geq1$ and set
\[
    k:=\min\{n,1+\lfloor\log_2 s\rfloor\}.
\]
For $j\in[k]:=\{1,\ldots,k\}$, define
\[
    E^{(j)}
    :=
    \sum_{\substack{x\in\{0,1\}^k\\x_j=1}}
        \kb{x,0^{n-k}},
\]
where $x,0^{n-k}$ is the concatenation of $x$ with $n-k$ zero bits. Thus
\begin{equation}
\label{eq:aaronson-structural-counts}
    \operatorname{rank}(E^{(j)})
    =
    \norm{E^{(j)}}_F^2
    =
    \operatorname{nnz}(E^{(j)})
    =
    2^{k-1}
    \leq s.
\end{equation}

Let $h:=\min\{k,T\}$, and partition the rounds into $h$ nonempty
consecutive blocks $I_1,\ldots,I_h$ whose lengths differ by at
most one.
In block $I_j$, use $E_t=E^{(j)}$ and
$\ell_t(u)=L|u-b_t|$, where the labels $b_t$ are independent
uniform bits, independent of the learner's randomness.
For any learner prediction $z_t:=\Tr(E_t\omega_t)\in[0,1]$,
\[
    \mathbb{E}[\ell_t(z_t)\mid z_t]
    =
    \frac L2\bigl(z_t+(1-z_t)\bigr)
    =
    \frac L2.
\]

Define the label imbalance in block $I_j$ by
\[
    S_j:=\sum_{t\in I_j}(2b_t-1).
\]
Choose $x^*\in\{0,1\}^k$ so that $x_j^*$ is a majority label
in $I_j$ for $j\leq h$, breaking ties arbitrarily and setting
unused bits to zero.
The single fixed state $\kb{x^*0^{n-k}}$ then has total loss
\[
    \frac L2
    \left(
        T-\sum_{j=1}^h|S_j|
    \right).
\]
Consequently,
\begin{equation}
\label{eq:aaronson-coordinate-lower-bound}
\begin{split}
    \mathbb{E}[R_T]
    &\geq
    \frac L2\sum_{j=1}^h\mathbb{E}|S_j|
    \\
    &\geq
    \frac{L}{2\sqrt{6}}\sqrt{Th}
    =
    \frac{L}{2\sqrt{6}}\min\{T,\sqrt{kT}\}.
\end{split}
\end{equation}
Here the expectation includes the labels and the learner's
randomness.
The second inequality uses $|I_j|\geq T/(2h)$ and the estimate
$\mathbb{E}|S|\geq\sqrt{q/3}$ for a sum $S$ of $q$ independent
uniform signs.
This estimate follows from H\"older's inequality,
$\mathbb{E}S^2=q$, and $\mathbb{E}S^4\leq3q^2$.
The argument applies to arbitrary learner states:
only their predictions being in $[0,1]$ was used.

Since
\[
    k
    =
    \Theta\!\left(
        \ln\bigl(1+\min\{s,d\}\bigr)
    \right),
\]
taking $s=r$, $s=B^2$ for $B\geq1$, and $s=\kappa$ proves
the three claims in these ranges, by
\eqref{eq:aaronson-structural-counts}.

For $0<B<1$, instead use
\[
    E_t=B\kb{10^{n-1}},
    \qquad
    \ell_t(u)=L|u-Bb_t|.
\]
Predictions then lie in $[0,B]$, and the same one-block
calculation, scaled by $B$, gives $\Omega(LB\sqrt T)$.
This implies the stated Frobenius bound because
$\ln(1+B^2)=\Theta(B^2)$ in this range.

Finally, averaging over the labels gives a deterministic label
sequence with the same lower bound in expectation over $\A$.

For the total matrix-entry sparsity bound, take $s=\kappa$.
Each constructed effect is diagonal with exactly $2^{k-1}$
nonzero entries, so
\[
    \operatorname{nnz}(E_t)=2^{k-1}\leq\kappa
    \qquad\text{for every }t.
\]
Combining \eqref{eq:aaronson-coordinate-lower-bound} with
$k=\Theta(\log(1+\min\{\kappa,d\}))$ proves
\eqref{eq:aaronson-sparsity-lower-bound}.
\end{proof}

\paragraph{Scope of the structural restrictions.}
The restrictions concern the effect $E_t$, not its complement
$\mathds{1}-E_t$.
Taking $k=n$ recovers
$\Omega(L\min\{T,\sqrt{nT}\})$ using projectors of rank $d/2$.
In this case
\[
    E^{(j)}=\frac{\mathds{1}-Z_j}{2},
\]
where $Z_j$ is Pauli $Z$ on qubit $j$ and the identity on the
other qubits.
Thus the same hard instance is one-sparse per matrix row and
has only two nonzero Pauli coefficients.
The sparsity parameter $\kappa$ above therefore counts total
matrix entries, not row sparsity or Pauli sparsity.
A zero structural budget forces $E_t=0$ and gives zero regret.

%% file: acks.tex
\section{Acknowledgment}

AB thanks Jamie Sikora, Sarvagya Upadhyay, Sumeet Khatri, and Ivona Brandi\'c for insightful discussions. 
AB's research is funded by BMIMI, BMWET, the State of Upper Austria, and the State of Tyrol within the COMET module Quantum Algorithm Engineering (FFG grant no. 923923) managed by Austrian Research Promotion Agency FFG.
This research was initiated during AB's visit to Fujitsu Research of America.

%% file: ai_disclosure.tex
\section{AI Disclosure}

ChatGPT 6 Astra was used to assist with exploring and validating proof strategies, as well as with manuscript preparation. The exposition of the proofs are largely due to the authors. The authors take full responsibility for the mathematical content, the final presentation, and any remaining errors.

%% file: ogd_regret.tex
\section{Upper bound on the regret for bounded effects via online gradient descent}\label{sec:regret-ogd}

We next adapt the Projected Online Gradient Descent (OGD) algorithm~\cite[Chapter 3]{hazan2016} to the problem of online learning of quantum states and analyze its regret in this setting. 

Recall that at each round $t \in [T]$, the learner selects a quantum state $\rho_t \in \mathrm{D}(\mathcal{X})$ and incurs the loss
\[
\ell_t(\tr(E_t \rho_t)),
\]
where the loss function $\ell_t$ and the measurement operator $E_t$ are chosen adversarially.
Here, $\mathrm{D}(\mathcal{X})$ denotes the set of quantum states on the Hilbert space $\mathcal{X}$, whose dimension scales exponentially with the number of qubits.

We analyze the following algorithm with $\Pi_{\textup{D}(\mathcal{X})}$ being the projection function that projects the input object onto the set $\textup{D}(\mathcal{X})$.

\begin{algorithm}[H]
\caption{OGD for online learning of quantum states}
\label{alg:ogd}
\begin{algorithmic}[1]
\State \textbf{Input:} $T$, $\mathcal{X}$, $\eta > 0$
\State Set $\rho_1 = \mathbb{1}/\operatorname{dim}(\mathcal{X})$.
\For{$t = 1, \ldots, T$}
    \State Consider the convex and $L$-Lipschitz loss function
    $\ell_t:\mathbb{R}\to\mathbb{R}$. Let $\ell_t'(x)$ be a
    (sub)derivative of $\ell_t$ with respect to $x$. Define
    \[
        \nabla_t \coloneqq
        \ell_t'(\operatorname{tr}(\rho_t E_t))\,E_t.
    \]
    \State Update the guess according to the online gradient descent rule:
    \begin{equation}\label{eq:pogd:projection}
        \rho_{t+1}
        \coloneqq
        \Pi_{\mathrm{D}(\mathcal{X})}
        \bigl(\rho_t-\eta\nabla_t\bigr).
    \end{equation}
\EndFor
\end{algorithmic}
\end{algorithm}

\subsection{Computing the Learner's Update via OGD}

The update rule of the Projected Online Gradient Descent (OGD) algorithm requires, at each iteration, computing the projection of a Hermitian operator onto the set of quantum states with respect to the Frobenius norm; see~\Cref{eq:pogd:projection} in~\Cref{alg:ogd}. Concretely, given an intermediate iterate, one must evaluate its closest point in $\mathrm{D}(\mathcal{X})$ under the Hilbert–Schmidt geometry.

An explicit and computationally efficient procedure for evaluating the successive updates (\Cref{eq:pogd:projection}) is provided in~\cite{gonccalves2016projected}. The derivation relies on the structural characterization which reduces the projection onto $\mathrm{D}(\mathcal{X})$ to a quadratic optimization problem over the probability simplex.

\subsection{A Tight Regret Analysis for OGD}

We now present a refined regret analysis of the Projected Online Gradient
Descent (OGD) algorithm in our setting which also significantly tightens the OGD analysis given in~\cite{yang2020revisiting}.
We begin by recalling a standard fact relating Lipschitz continuity and
bounded gradients, stated here in a slightly generalized form for
completeness.

\medskip

\begin{proposition}\label{prop:lipschitz_gradient_bound}
Let $S$ be an open subset of a normed vector space $(\mathcal{V},\|\cdot\|)$,
and let $f : S \to \mathbb{R}$ be a convex and $L$-Lipschitz function on $S$,
that is,
\[
|f(X) - f(Y)| \le L \|X-Y\|
\quad \text{for all } X,Y \in S.
\]
Then, for every $X \in S$ at which $f$ is differentiable,
\[
\|\nabla f(X)\|_* \le L,
\]
where $\|\cdot\|_*$ denotes the dual norm.
\end{proposition}

\begin{proof}
Fix $X \in S$. Since $S$ is open, there exists $\alpha > 0$ such that
$\{Y : \|Y-X\|\le \alpha\} \subseteq S$.
For any $Z$ with $\|Z\|\le \alpha$, convexity implies
\[
f(X+Z) \ge f(X) + \langle \nabla f(X), Z\rangle .
\]
Using the Lipschitz property,
\[
|\langle \nabla f(X), Z\rangle|
\le |f(X+Z)-f(X)|
\le L\|Z\|.
\]
Taking the supremum over all $\|Z\|\le \alpha$ and using the definition
of the dual norm yields
\[
\alpha \|\nabla f(X)\|_* 
= \max_{\|Z\|\le \alpha} \langle \nabla f(X), Z\rangle
\le \alpha L,
\]
which implies $\|\nabla f(X)\|_* \le L$.
\end{proof}

\medskip

We next record the standard OGD one-step progress inequality specialized
to our setting.

\begin{claim}\label{fact:ogdBounds}
Let $\ell_t:\mathbb{R}\to\mathbb{R}$ be convex and $L$-Lipschitz.
Let $\ell_t'(x)$ be a sub-derivative of $\ell_t$ at $x$, and define
\[
\nabla_t \coloneqq \ell_t'\!\big(\operatorname{tr}(E_t \rho_t)\big)\,E_t.
\]
Then, for every $U \in \mathrm{D}(\mathcal{X})$,
\begin{align*}
\hspace{-0.4cm}
\eta\!\left(\ell_t(\operatorname{tr}(E_t \rho_t))
     - \ell_t(\operatorname{tr}(E_t U))\right)
&\le \eta\, \operatorname{tr}\!\big(\nabla_t(\rho_t-U)\big) \\
&\le \tfrac12 \|\rho_t-U\|_F^2 - \tfrac12 \|\rho_{t+1}-U\|_F^2 + \tfrac{\eta^2}{2}\|\nabla_t\|_F^2 .
\end{align*}
\end{claim}

\begin{proof}
By the OGD update rule,
\[
\rho_{t+1}
= \Pi_{\mathrm{D}(\mathcal{X})}
  \big(\rho_t - \eta \nabla_t\big).
\]
Non-expansiveness of Euclidean projection implies
\[
\|\rho_{t+1}-U\|_F^2
\le \|\rho_t - \eta\nabla_t - U\|_F^2 .
\]
Expanding the right-hand side and rearranging yields
\[
\|\rho_{t+1}-U\|_F^2
\le \|\rho_t-U\|_F^2
+ \eta^2\|\nabla_t\|_F^2
-2\eta\,\operatorname{tr}\!\big(\nabla_t(\rho_t-U)\big),
\]
which gives the claimed inequality.
\end{proof}

\medskip

We are now ready to state the main regret guarantee.

\smallskip

\begin{theorem}\label{theorem:PGDRegret}
Suppose the loss at each round $t \in [T]$ is of the form
$\ell_t(\operatorname{tr}(E_t \rho_t))$, where
$\ell_t : \mathbb{R} \to \mathbb{R}$ is convex and $L$-Lipschitz.
Assume further that $\|E_t\|_F \le B$ for all $t \in [T]$.
Then the regret of OGD satisfies
\[
R_T \le D L B\sqrt{T},
\]
where $D$ denotes the diameter of $\mathrm{D}(\mathcal{X})$
with respect to the Frobenius norm.
\end{theorem}

\begin{proof}
Fix $U \in \mathrm{D}(\mathcal{X})$.
Summing the inequality in \Cref{fact:ogdBounds} over $t=1,\dots,T$
and telescoping gives
\begin{align*}
R_T(U)
&\le \frac{1}{2\eta}\|\rho_1-U\|_F^2
 + \frac{\eta}{2}\sum_{t=1}^T \|\nabla_t\|_F^2 .
\end{align*}
Since $\mathrm{D}(\mathcal{X})$ has Frobenius diameter at most $D$,
we have $\|\rho_1-U\|_F \le D$.
Moreover, by \Cref{prop:lipschitz_gradient_bound},
$|\ell_t'|\le L$, and thus
\[
\|\nabla_t\|_F
= |\ell_t'(\operatorname{tr}(E_t \rho_t))|
  \,\|E_t\|_F
\le L B .
\]
Therefore,
\[
R_T(U)
\le \frac{D^2}{2\eta}
 + \frac{\eta T}{2} L^2 B^2 .
\]
Optimizing by setting $\eta = \frac{D}{LB\sqrt{T}}$
yields
\[
R_T \le D L B \sqrt{T}.
\]

Finally, note that this bound is independent of the number of qubits.
Since the Frobenius diameter of the set of density operators is at most
$D=2$, the bound simplifies to $R_T \le 2 L B \sqrt{T}$.
\end{proof}

\begin{remark}\label{rem:why-frob-misses}
    [Why OGD misses the logarithm] \Cref{alg:ogd} fixes \(p=q=2\).  Under a rank-\(r\) constraint, this pays \(\norm E_{S_2}\leq\sqrt r\).  Schatten mirror descent chooses \(q\) just above one, so its dual exponent \(p\) is of order \(\ln r\).  The rank factor then becomes \(r^{1/p}=O(1)\), while the loss of strong convexity costs only \(\sqrt{p-1}=O(\sqrt{\ln r})\).  This is the entire source of the improvement.
\end{remark}

%% file: sparse_lower.tex
\section{Lower bound on the regret for Pauli-sparse effects}\label{sec:regret-sparse-effects}

We now show that the \(\Omega(\sqrt{nT})\) regret lower bound already holds even when every measurement is sparse in the Pauli basis, i.e., every measurement has at most \(\kappa\) non-zero Pauli coefficients.

Consider the \(n\) single-qubit measurements
\[
E_i:=\frac{\mathbb{1}+Z_i}{2},\qquad i\in[n],
\]
where \(Z_i\) denotes the Pauli \(Z\) operator acting on qubit \(i\) and identity elsewhere. Each \(E_i\) is a valid two-outcome effect and is \(2\)-sparse in the Pauli basis.

Consider the \(L_1\) loss
\[
\ell_t(z)=|z-b_t|,
\]
where \(b_t\in\{0,1\}\) is a random and independently chosen bit. 


At each round $t$ in which \(E_i\) is presented, let $\rho_t$ be the learner's action with
\[
p_t=\operatorname{Tr}(E_i\rho_t)\in[0,1]
\]
being the learner's prediction. Since \(b_t\) is an independent fair bit,
\[
\mathbb{E}\!\left[\ell_t(p_t)\right]
=
\frac{1}{2}(\lvert p -1 \rvert + \lvert p \rvert)
=
\frac{1}{2}.
\]
Hence the learner's expected cumulative loss is
\[
\mathbb{E}\!\left[\sum_{t=1}^T\ell_t(p_t)\right]
=\frac{T}{2}.
\]

It remains to upper-bound the loss of the best fixed quantum state. 
The state space contains every computational-basis state
\(\lvert x_1\cdots x_n\rangle\), \(x_i\in\{0,1\}\). Moreover,
\[
\operatorname{Tr}\!\left(E_i|x\rangle\langle x|\right)
=
\begin{cases}
1,&x_i=0,\\
0,&x_i=1.
\end{cases}
\]
Thus a single fixed computational-basis state can independently choose either prediction \(0\) or \(1\) for every coordinate \(i\). If coordinate \(i\) is queried \(m_i\) times, its minimum cumulative \(L_1\) loss is
\[
\min\{N_i^0,N_i^1\}
=
\frac 12\left(m_i-|S_i|\right),
\]
where \(N_i^0,N_i^1\) are the numbers of zeros and ones among the corresponding \(b_t\)'s with $S_i \coloneqq N_i^1 - N_i^0$. Therefore
\[
\min_{\phi}\sum_{t=1}^T
\ell_t\!\left(\operatorname{Tr}(E_t\phi)\right)
=
\frac 12\left(T-\sum_{i=1}^n|S_i|\right).
\]

Thus, the expected regret is given by
\[
\mathbb{E}[R_T]
=
\frac 12\sum_{i=1}^n\mathbb{E}|S_i|.
\]
For a sum \(S_i\) of \(m_i\) independent Rademacher random variables, the standard \(p=1\) Khintchine inequality gives
\[
\mathbb{E}|S_i|
\ge
\sqrt{\frac{m_i}{2}}.
\]
Hence
\[
\mathbb{E}[R_T]
\ge
\frac{1}{2\sqrt2}\sum_{i=1}^n\sqrt{m_i}.
\]

Assuming \(T\ge n\), then, with almost balanced choice of $m_i$,
\[
m_i\ge \frac{T}{2n}
\]
for every \(i\). Thus
\[
\sum_{i=1}^n\sqrt{m_i}
\ge
n\sqrt{\frac{T}{2n}}
=
\sqrt{\frac{nT}{2}},
\]
which yields
\[
\mathbb{E}[R_T]\ge \frac{1}{4}\sqrt{nT}.
\]

Since the above is a randomized choice of measurements and losses, it implies the existence of a deterministic adversarial sequence with at least the same expected regret against the learner. Therefore,
\[
R_T \ge \frac 14\sqrt{nT}.
\]

In particular, for every \(\kappa\ge2\), restricting the adversary to \(\kappa\)-sparse Pauli-basis measurements does not remove the \(\Omega(\sqrt{nT})\) lower bound in the regime \(T\ge n\).